\documentclass[twocolumn,amsmath,amssymb,superscriptaddress]{revtex4}
\usepackage{dcolumn}
\usepackage{bm}
\usepackage{amsthm}
\newtheorem{example}{Example}
\newtheorem{proposition}{Proposition}
\usepackage[pdftex]{graphicx}
\graphicspath{{figures/pdf/}} \DeclareGraphicsExtensions{.pdf}
\def\vec#1{{\bm #1}}

\def\ket#1{| #1 \rangle}
\def\bra#1{\langle #1 |}
\def\ip#1#2{\langle #1 | #2 \rangle}

\def\norm#1{\| #1 \|}

\def\RR{\mathbb{R}}
\def\diag{\operatorname{diag}}

\def\dim{\operatorname{dim}}

\def\Tr{\operatorname{Tr}}

\def\A{\mathop{\bf A}\nolimits}
\def\Lb{\mathop{\bf L}\nolimits}
\def\Db{\mathop{\bf D}\nolimits}

\def\D{\mathcal{D}}
\def\E{\mathcal{E}}

\def\H{\mathcal{H}}

\def\L{\mathcal{L}}

\def\ONE{\mathbb{I}}

\def\SU{\mathbb{SU}}

\def\BB{\mathfrak{B}}
\def\DD{\mathfrak{D}}

\def\HS{\mathop{\rm HS}}
\def\TR{\mathop{\rm TR}}
\def\ss{\rm ss}
\begin{document}
\title{On the contractivity of the Hilbert-Schmidt distance under open system
dynamics}
\author{Xiaoting Wang}\email{xw233@cam.ac.uk}
\affiliation{Department of Applied Maths and Theoretical Physics,
             University of Cambridge, Wilberforce Road, Cambridge, CB3 0WA, UK}
\author{S.~G.~Schirmer}\email{sgs29@cam.ac.uk}
\affiliation{Department of Applied Maths and Theoretical Physics,
             University of Cambridge, Wilberforce Road, Cambridge, CB3 0WA, UK}
\date{\today}

\begin{abstract}
It is shown that the Hilbert-Schmidt (HS) norm and distance, unlike the
trace norm and distance, are generally not contractive for open quantum
systems under Lindblad dynamics.  Necessary and sufficient conditions
for contractivity of the HS norm and distance are given, and explicit
criteria in terms of the Lindblad operators are derived.  It is also
shown that the requirements for contractivity of the HS distance are
strictly weaker than those for the HS norm, although simulations suggest
that non-contractivity is the typical case, i.e., that systems for which
the HS distance between quantum states is monotonically decreasing are
exceptional for $N>2$, in contrast to the case $N=2$ where it is always
monotonically decreasing.
\end{abstract}

\pacs{Add PACS} 

\maketitle

\section{Introduction}

The trace norm and induced trace distance play an important role in
quantum information theory.  One of its most important features,
widely used in research papers and popular
textbooks~\cite{Nielson,Breuer}, is its contractivity for
trace-preserving quantum evolution, first proved by
Ruskai~\cite{Ruskai}.  It was once conjectured that contractivity
would extend to other more intuitive norms such that the
Hilbert-Schmidt (HS) norm~\cite{Vedral,Witte}, but flaws in the
original argument were soon discovered and an explicit
counter-example of a trace-preserving map with non-contractive HS
norm was provided by Ozawa~\cite{Ozawa}. The more general question
of when a positive trace-preserving [PTP] map between matrix spaces
is contractive with respect to the $p$-norm for $p>1$ was recently
considered in~\cite{Garcia}, where it was shown that PTP maps are
contractive for $p>1$ in general if only if they are unital.  This
does not answer the question, however, when a PTP map defined on a
subset of a matrix space is contractive with respect to a particular
$p$-norm such as the HS norm, and contractivity depends on the
subspace.  In particular, the HS norm restricted to the trace-zero
hyperplane in the space of Hermitian matrices may be contractive
while it is not on the whole matrix space.

Contractivity of the distance between quantum states in the sense
that the distance between two quantum states is monotonically
decreasing, is of particular interest in the area of open system
dynamics.  We consider under what conditions the evolution of an
open quantum system subject to semigroup dynamics governed by a
Lindblad master equation is contractive with respect to the HS
distance.  It can be shown that for $N=2$ both the HS norm and
distance are always contractive, while for $N>2$ both are
contractive only for small subsets of open systems, with the set of
open systems for which the HS norm is contractive being strictly
smaller than the set of open systems for which the HS distance is
contractive. Aside from the relevance to quantum information, e.g.,
in the construction of entanglement measures, one important
implication of the non-contractivity of the HS norm is in quantum
control, where the HS distance between quantum trajectories is a
common choice of a Lyapunov function, e.g., for local optimal
control~\cite{Mirrahimi2005,altafini,Wang-Schirmer,Ticozzi}.  As
monotonicity is a prerequisite for a Lyapunov function, it means
that the HS distance is generally not a suitable candidate for a
Lyapunov function for open systems.

\section{Hilbert Space Norms and Contractivity}

Let $\H$ be a Hilbert space with $\dim\H=N$ and let $\BB[\H]$ be the
bounded operators on $\H$, and $\DD[\H]$ be the trace-$1$ positive
operators on $\H$.  It is not difficult to check that
\begin{equation}
  \norm{A}_{\TR} = \frac{1}{2} \Tr \sqrt{A^\dag A}
\end{equation}
is well-defined for any $A\in\BB[\H]$ as $A^\dag A$ is positive, and
satisfies the axioms of a norm, and induces a metric, the trace
distance, in the usual way $d_{\TR}(A,B)=\norm{A-B}_{\TR}$.

Among the nice features of the trace norm is contractivity under
trace-preserving maps~\cite{Ruskai}, i.e., given two density
operators $\rho_1,\rho_2\in\DD[\H]$ and a map $\E:\BB[\H]\to\BB[\H]$
with $\Tr(\E(A))=\Tr(A)$, then
\begin{equation}
   d_{\TR}(\E(\rho_1),\E(\rho_2)) \le d_{\TR}(\rho_1,\rho_2).
\end{equation}
This implies in particular that any (super)operator that maps
density matrices to density matrices cannot increase the trace
distance between two states.  Furthermore, if $\E_t$ is a quantum
dynamical semi-group then $d_{\TR}(\E_t(\rho_1),\E_t(\rho_2))$ is
monotonically decreasing, although monotonicity is not always
strict. Indeed, in the special case of unitary evolution the trace
distance remains constant,
$d_{\TR}(\E_t(\rho_1),\E_t(\rho_2))\!=\!d_{\TR}(\rho_1,\rho_2)$ for
all $t$.


One drawback of the trace distance, however, is that it is not the
most intuitive distance measure for quantum states.  A more natural
choice is the Hilbert Schmidt (HS) distance
$d_{\HS}(A,B)=\norm{A-B}_{\HS}$ where
$\norm{A}_{\HS}=\sqrt{\ip{A}{A}}$ is the norm induced by the HS
inner product
\begin{equation}
  \ip{A}{B}_{\HS} = \Tr (A^\dag B).
\end{equation}
For Hermitian operators $A=A^\dag$, $B=B^\dag$ the HS inner product
and distance simplify to $\ip{A}{A}_{\HS}=\Tr(A^2)$ and
\begin{equation}
  d_{\HS}(A,B) = \sqrt{\Tr[(A-B)^2]}.
\end{equation}
If we choose an orthonormal basis $\{\sigma_k\}_{k=1}^{N^2}$ for the
Hermitian operators on $\H$ then the coordinate vector
$\tilde{\vec{a}}=(a_k)_{k=1}^{N^2}$ with $a_k=\Tr(\sigma_k A)$ for
any Hermitian operator $A$ on $\H$ is a vector in $\RR^{N^2}$, and
noting that $\Tr(\sigma_k\sigma_\ell)=\delta_{k\ell}$ by
orthonormality, shows that the HS inner product reduces to the
standard Euclidean inner product in $\RR^{N^2}$
\begin{equation*}
  \ip{A}{A}_{\HS}
  = \sum_{k,\ell=1}^{N^2} a_k a_\ell \Tr(\sigma_k\sigma_\ell)
  = \sum_{k=1}^{N^2} a_k^2 = \ip{\tilde{\vec{a}}}{\tilde{\vec{a}}}.
\end{equation*}
If we choose an orthonormal basis with the last basis vector
$\sigma_{N^2}=\frac{1}{\sqrt{N}}\ONE$, where $\ONE$ is the identity
on $\H$, then all Hermitian operators with constant trace $\Tr(A)=c$
correspond to hyperplanes in $\RR^{N^2}$ with
$a_{N^2}=\frac{1}{\sqrt{N}}c$.  Therefore, if we are only interested
in Hermitian operators with constant trace, such as the class of
density operators on $\H$, it suffices to consider the reduced
coordinate vector $\vec{a}=(a_1,\ldots,a_{N^2-1})$ and
\begin{equation}
   \ip{A}{A}_{\HS} = N^{-1} \Tr(A) + \ip{\vec{a}}{\vec{a}}.
\end{equation}
Furthermore, if $A,B$ are two Hermitian operators with
$\Tr(A)=\Tr(B)$ then $\Tr(A-B)=0$ and thus
\begin{equation}
  d_{\HS}(A,B) = \sqrt{\ip{\vec{a}-\vec{b}}{\vec{a}-\vec{b}}}
               = \norm{\vec{a}-\vec{b}},
\end{equation}
i.e., the HS distance between two Hermitian operators of the same
trace class, is simply the Euclidean distance of their associated
reduced (real) coordinate vectors, which generalize the Bloch vector
for $N=2$. This equivalence of HS distance between density operators
and the Euclidean distance between their reduced coordinate vectors
makes the HS distance a very intuitive and useful distance measure.

\section{Criteria for Monotonicity of HS Norm and Distance}

For $N=2$, i.e., a single qubit, it is easy to show that the trace
and HS distance agree up to a constant factor, and hence
contractivity of one implies contractivity of the other.  Let
$\rho_1$ and $\rho_2$ be two qubit density operators and let
$\vec{r}_1$ and $\vec{r}_2$ be their respective coordinate vectors
in $\RR^3$ with respect to the (orthonormal) basis
$\{\sigma_k\}_{k=1}^4$ defined above.  Then $\rho_j = \sum_{k=1}^3
r_{jk} \sigma_k + \frac{1}{2}\ONE$ for $j=1,2$ and setting
$\vec{\sigma}=(\sigma_1,\sigma_2,\sigma_3)$
\begin{equation}
  \rho_1-\rho_2 = \sum_{k=1}^3 (r_{1k}-r_{2k})\sigma_k
                = (\vec{r}_1-\vec{r}_2) \cdot \vec{\sigma}.
\end{equation}
Choosing $\vec{\sigma}$ to be, e.g., the standard (normalized) Pauli
basis, it is easy to verify that the eigenvalues of $\rho_1-\rho_2$
are $\pm \frac{1}{\sqrt{2}} \norm{\vec{r}_1-\vec{r}_2}$, and noting
that for Hermitian matrices $\Tr |A|$ is the sum of the absolute
values of the eigenvalues of $A$, we have~\footnote{If we used the
unnormalized Pauli matrices as a basis we would obtain a factor of
$\frac{1}{2}$ instead, but using an orthonormal basis the Bloch
sphere for a single qubit has radius $\sqrt{\frac{1}{2}}$ instead of
$1$, whence we only a factor of $\sqrt{\frac{1}{2}}$.}
\begin{equation}
\begin{split}
  d_{\TR}(\rho_1,\rho_2) &= \frac{1}{2} \Tr |\rho_1-\rho_2| \\
                         &= \frac{1}{2}
\frac{2}{\sqrt{2}}\norm{\vec{r}_1-\vec{r}_2}
                         = \frac{1}{\sqrt{2}} d_{\HS}(\rho_1,\rho_2).
\end{split}
\end{equation}
Hence, the HS distance for a single qubit is monotonically
decreasing under completely positive maps and quantum semigroup
dynamics.  This argument does not generalize, however, as the
relationship between the eigenvalues of a positive map and the
Euclidean norm of the coordinate vector that was used is very
specific to $N=2$.

For $N>2$ we still have $\norm{A}_{\HS}\le\norm{A}_{\TR}$ from basic
functional analysis~\cite{Yosida}.  Thus the trace norm provides an
upper bound on the Hilbert-Schmidt norm.  Applied to open quantum
systems subject to semi-group dynamics governed by a Lindblad master
equation
\begin{equation}
 \label{eq:LME}
  \dot\rho(t) = -i[H,\rho(t)] + \L_D\rho(t),
\end{equation}
with $\L_D\rho(t)=\sum_{d}\D[V_d]\rho(t)$, $V_d\in\BB[\H]$ and
\begin{equation}
  \label{eq:D}
  \D[V_d] \rho(t) = V_d \rho(t) V_d^\dagger
                   - \frac{1}{2}(V_d^\dagger V_d \rho(t)
           + \rho(t) V_d^\dagger V_d),
\end{equation}
this means that if the trace distance between any two quantum states
goes to zero for $t\to\infty$, for instance, then the HS distance
must go to zero as well, but the convergence need not be monotonic.

Expanding the density operator $\rho$ and super-operators $\L_H$ and
$\L_D$ with respect to an orthonormal basis for the (trace-zero)
Hermitian matrices on $\H$, it is easy to show that the master
equation~(\ref{eq:LME}) becomes an affine-linear equation for the
coordinate vector $\vec{r}\in\RR^{N^2-1}$
\begin{equation}
  \label{eq:bloch}
  \dot{\vec{r}}(t) = \A \vec{r}(t) + \vec{c},
\end{equation}
where $\A$ is a $(N^2-1)\times(N^2-1)$ real matrix and
$\vec{c}\in\RR^{N^2-1}$ can be computed from the Hamiltonian and
Lindblad generators (See Appendix B).

\begin{proposition}
\label{prop:1} A necessary and sufficient condition for the HS norm
$\norm{\rho(t)}_2$ of any quantum state $\rho(t)$ to be
monotonically decreasing under the open system
dynamics~(\ref{eq:bloch}) is that the Bloch equation is linear,
i.e., $\vec{c}=\vec{0}$.
\end{proposition}

\begin{proof}
If $T$ is a PTP map acting on the Hermitian operators on a
finite-dimensional Hilbert space then
$\norm{T\rho}_2\le\norm{\rho}_2$ if and only if $T$ is
unital~\cite{Garcia}.  For our open system dynamics, the evolution
is unital, i.e., preserves the identity $\ONE$, if and only if the
Bloch equation is homogeneous, i.e., $\vec{c}=0$, as the identity
$\ONE$ is mapped to $\vec{r}=\vec{0}$ in the trace-one hyperplane
the density operators live in.
\end{proof}

\begin{proposition}
\label{prop:2} A necessary and sufficient condition for the HS
distance $\norm{\rho_1(t)-\rho_2(t)}$ between quantum states to be
monotonically decreasing under the dynamics~(\ref{eq:bloch}) is that
the symmetric part $\A+\A^T$ of the evolution operator $\A$ be
negative definite, i.e., have no positive eigenvalues.
\end{proposition}

\begin{proof}
Denote the set of all physical coordinate vectors as
$\DD_{\RR}[\H]$, corresponding to $\DD[\H]$.  We can easily see that
the eigenvalues of $\A$ must all have non-positive real parts since
the dynamical system is invariant on the compact set $\DD_{\RR}[\H]$.
Moreover, if $\rho_1(0)=\rho_1$ and $\rho_2(0)=\rho_2$ are two
initial states with coordinate vectors $\vec{r}_1$ and $\vec{r}_2$,
respectively, then $\vec{\Delta}(t)=\vec{r}_1(t)-\vec{r}_2(t)$
satisfies the linear equation
\begin{equation}
  \label{eq:bloch2}
  \dot{\vec{\Delta}}(t) = \A \vec{\Delta}(t) \quad \Rightarrow \quad
  \vec{\Delta}(t)       = e^{t\A}\vec{\Delta}(0)
\end{equation}
and thus we have
\begin{align*}
 \frac{d}{dt} d_{\HS}^2(\vec{r}_1(t),\vec{r}_2(t))
 &= \ip{\dot{\vec{\Delta}}(t)}{\vec{\Delta}(t)}
   +\ip{\vec{\Delta}(t)}{\dot{\vec{\Delta}}(t)} \\
 &= (\A\vec{\Delta}(t))^T\vec{\Delta}(t) +
    \vec{\Delta}(t) A\vec{\Delta}(t) \\
 &= \vec{\Delta}(t)^T ({\A}^T+\A) \vec{\Delta}(t).
\end{align*}
As $\A+\A^T$ is real symmetric, its eigenvalues are real, and the HS
distance will be monotonically decreasing, if and only if $\A+\A^T$
has only non-positive eigenvalues.
\end{proof}

Note that although we know that the eigenvalues of $\A$ have
non-positive real parts, this does not imply that $\A+\A^T$ has
non-positive (real) eigenvalues in general.

\begin{example}
Consider the system $\dot\rho(t)=\D[V]\rho$ with the simple Lindblad
generator
\begin{equation*}
  V=\begin{pmatrix} 1 & 1 & 1\\ 0 & 1 & 1 \\ 0 & 0 & 1 \end{pmatrix}.
\end{equation*}
Choosing $\{\sigma_k:k=1,\ldots,8\}$ to be the standard orthonormal
basis for the trace-zero Hermitian matrices~(\ref{eq:pauliN}) for
$N=3$, and $\sigma_{N^2}=\frac{1}{\sqrt{3}}\ONE$, we obtain
\begin{equation*}
\A = \begin{pmatrix}
    -1 & 0 & 0 & 1 & \frac{1}{3}\sqrt{3} & 0 & \frac{1}{2} & 1\\
    0 & -\frac{1}{2} & 0 & 0 & 0 & \frac{1}{2} & 0 & 0\\
    0 & -1 & -1 & 0 & 0 & \frac{1}{2} & 0 & 0\\
    -1& 0 & 0 & -\frac{1}{2} & -\frac{2}{3}\sqrt{3} & 0  & 0 & \frac{3}{2}\\
    0 & 0 & 0 & 0 & -2 & 0 & \frac{1}{2}\sqrt{3} &\sqrt{3}\\
    0 & \frac{1}{2} &-\frac{1}{2} & 0 & 0& -\frac{3}{2} & 0 & 0\\
    -\frac{1}{2} & 0 & 0 & -1 & -\frac{1}{2}\sqrt{3} & 0 & -1 & \frac{1}{2}\\
    1 & 0 & 0 & -\frac{1}{2} &-\frac{1}{3}\sqrt{3} & 0 &-\frac{1}{2}&
-\frac{3}{2}
    \end{pmatrix}
\end{equation*}
as well as $\vec{c}=\frac{\sqrt{2}}{3}(1,0,0,1,\sqrt{3},0,0,-1)^T$.
It is easy to check that $\A$ is invertible, its eigenvalues have
negative real parts, and the system has a unique steady state
$\vec{r}_{\ss}=-\A^{-1}\vec{c}$ corresponding to
\begin{equation*}
  \rho_{\ss} = \frac{1}{5}
  \begin{pmatrix}
  2 & -1 & 0 \\
 -1 &  2 &-1 \\
  0 & -1 & 1
  \end{pmatrix}.
\end{equation*}
However, $[A,A^T]\neq 0$, i.e., $\A$ is not normal and $\A+\A^T$ has
a \emph{positive} eigenvalue $\gamma=0.1914$ with eigenvector
$\vec{v}$. Hence, setting $\vec{\Delta}(0)=\alpha\vec{v}$ with
$\alpha>0$ chosen such that
$\vec{r}(0)=\vec{r}_{\ss}+\alpha\vec{v}\in\DD_\RR[\H]$ gives a
trajectory for which
\begin{equation}
  \frac{d}{dt}d^2(0)=\alpha^2\vec{v}^T (A^T+A) \vec{v}
  =\gamma\alpha^2\norm{\vec{v}}^2>0,
\end{equation}
and thus the distance from the steady state increases at least
initially.  Indeed Fig.~\ref{fig1} shows that for
$\vec{\Delta}(0)=\frac{1}{9}\vec{v}$ and $\vec{s}(0)\in\DD_\RR[\H]$
the HS distance $d_{\HS}(t)=\norm{\vec{\Delta}(t)}$ eventually
converges to $0$---as it must as $\A$ has no eigenvalues with real
part $0$---but it increases initially.
\end{example}

\begin{figure}
\includegraphics[width=\columnwidth]{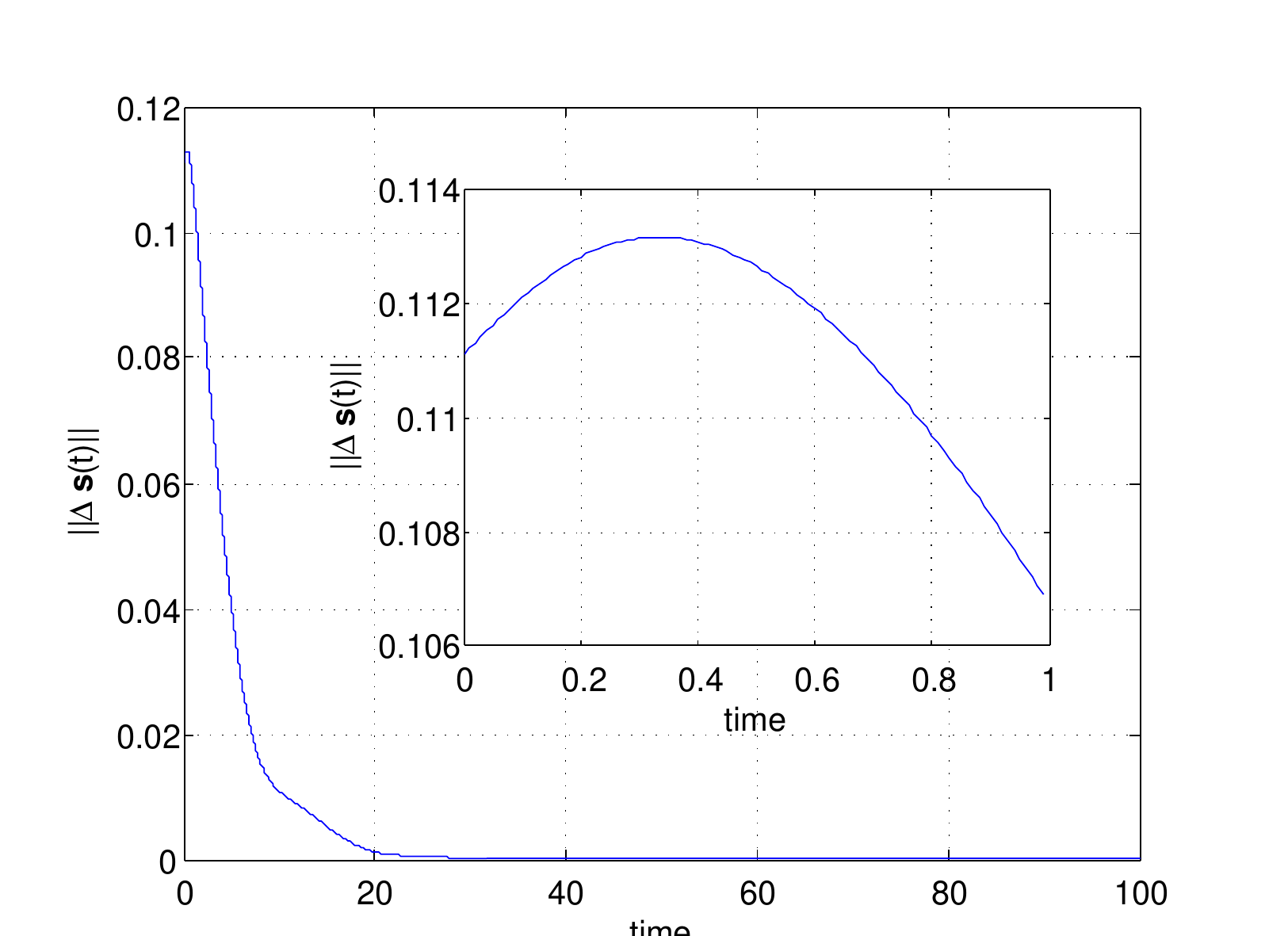}
\caption{Non-monotonic convergence to steady state.} \label{fig1}
\end{figure}

Whenever $\A+\A^T$ has a positive eigenvalue $\gamma$ with
eigenvalue $\vec{v}$ then the distance between any two initial
states $\rho_1(0)$ and $\rho_2(0)$, whose corresponding real
coordinate vectors satisfy $\vec{r}_1-\vec{r}_2=\alpha \vec{v}$,
will increase at least initially.  Furthermore, if the system has a
steady state $\rho_{\ss}$ in the interior of $\DD[\H]$ then there
are initial states that do not converge to the steady state
monotonically with respect to the HS distance.  This is easy to see
since for a point in the interior, it is always possible to choose
$\alpha>0$ such that $\vec{s}(0)=\vec{s}_{\ss}+\alpha\vec{v} \in
\DD_\RR[\H]$, where $\DD_\RR[\H]$ is the subset of $\RR^{N^2-1}$
corresponding to physical states, for any $\vec{v}$.

\section{Sufficient Conditions on Lindblad operators for Monotonicity}

A necessary and sufficient condition for the HS norm of a quantum
state to be monotonically decreasing is that $\L_D(\ONE)=0$.
Inserting this into the LME~(\ref{eq:LME}) leads to $\sum_d
[V_d,V_d^\dag]=0$, which gives the following explicit result:

\begin{proposition}
The HS norm of any quantum state is monotonically decreasing under
the open system dynamics~(\ref{eq:bloch}) if and only if the
Lindblad operators $V_d$ satisfy $\sum_d [V_d,V_d^\dag] =0$.
\end{proposition}

We know that contractivity of the HS norm is a sufficient condition
for contractivity of the HS distance, but we can derive other
sufficient conditions.

\begin{proposition}
The distance between any two quantum states is monotonically
decreasing under the open system dynamics~(\ref{eq:bloch}) if $\A$
is normal.
\end{proposition}

\begin{proof}
If $\A$ is normal, i.e., $[\A,\A^T]=0$, then there exists a unitary
transformation $\mathbf{U}\in\SU(N^2-1)$ such that
$\A=\mathbf{UDU}^\dag$, where $\mathbf{D}$ is a diagonal matrix.
Hence, noting that $\A$ is real, $\A^T=\A^\dag=\mathbf{UD^\dag
U^\dag}$, and thus $\A+\A^T=\mathbf{U(D+D^\dag)U^\dag}$ shows that
the eigenvalues of $\A+\A^T$ are twice the real parts of those of
$\A$, i.e., non-positive.
\end{proof}

For simple Lindblad equations we can further show that normality of
the Lindblad operators is a sufficient condition for normality of
$\A$.

\begin{proposition}
For a system governed by the purely dissipative LME
$\dot\rho(t)=\D[V]\rho(t)$ the superoperator $\A+\A^T$ is normal if
$V$ is normal.
\end{proposition}

\begin{proof}
Let $a_{\ell k}$ be the $(\ell,k)$th component of $\A_V$.  From the
definition of the Bloch equation it follows
\begin{equation}
  a_{\ell j} = \Tr(V^\dagger \sigma_\ell V\sigma_j)
               -\frac{1}{2}\Tr(V^\dag V \{\sigma_\ell,\sigma_j\}),
\end{equation}
where ${A,B}=AB+BA$ is the anticommutator, and thus the $(\ell,k)$th
component of $\A_V \A_V^T$ is equal to $\sum_j a_{\ell j} a_{kj}$
and the $(\ell,k)$th component of $\A_V^T \A_V$ is $\sum_j
a_{j\ell}a_{jk}$. If $H_1$ and $H_2$ are two Hermitian matrices and
at least one has zero trace, we have the identity
\begin{align*}
  \sum_j \Tr(H_1\sigma_j)\Tr(H_2\sigma_j) = \Tr(H_1H_2),
\end{align*}
and thus normality of $A_V$ is equivalent to
\begin{align*}
 0=& \Tr(\sigma_\ell V \sigma_k V^\dag [V,V^\dag])
    +\Tr(\sigma_k V \sigma_\ell V^\dag [V,V^\dag])\\
   &+\Tr(V \sigma_\ell V^\dag \sigma_k [V,V^\dag])
    +\Tr(V \sigma_k V^\dag \sigma_\ell [V,V^\dag])\\
    &+\Tr(V V^\dag \sigma_\ell V V^\dag\sigma_k )-
    \Tr(V^\dag V \sigma_\ell V^\dag V \sigma_k ),
\end{align*}
for all $k,\ell$, which is satisfied if $[V,V^\dag]=0$.
\end{proof}

If there are multiple Hamiltonian and Lindblad terms then by
linearity of the master equation, the superoperator splits, i.e.,
$\A=\A_H+\sum_d \A_{V_d}$, where $\A_H$ is associated with the
Hamiltonian dynamics and $\A_{V_d}$ corresponds to the decoherence
operator $V_d$.  It is easy to see that $\A_H$ is
real-antisymmetric, and thus $\A_H+\A_H^T=0$.  Since the sum of
negative semi-definite matrices is negative semi-definite, $\A+\A^T$
is negative semi-definite if $\A_{V_d}+\A_{V_d}^T$ is negative
semi-definite for all $d$, and the latter is the case if $\A_{V_d}$
is normal.  Thus we can conclude that for a system governed by the
LME~(\ref{eq:LME}) the superoperator $\A+\A^T$ is negative definite
if all Lindblad operators $V_d$ are normal.  This is of course
consistent with the previous observation that $\sum_d
[V_d,V_d^\dag]=0$ implies contractivity of the HS norm, and hence
contractivity of the HS distance.  We might thus conjecture that the
two sufficient conditions for contractivity of the HS distance are
equivalent, but this is not the case.

\begin{example}
Consider the a three-level system with $H=0$ and two Lindblad
operators
\begin{equation}
  V_1 = \begin{pmatrix} 0 & 0 & 1\\ 0 & 0 & 0 \\ 0 & 0 & 0 \end{pmatrix},
  \;
  V_2 = \begin{pmatrix} 0 & 0 & 0\\ 1 & 0 & 0 \\ 0 & 1 & 0 \end{pmatrix}.
\end{equation}
It is easy to check that $[V_d,V_d^\dag]\neq 0$ for $d=1,2$ but
$\sum_{d=1,2} [V_d,V_d^\dag]=0$ and thus $\vec{c}=\vec{0}$.  Thus
the evolution is unital and both the HS norm and distance are
contractive, but the superoperator $\A$ is \emph{not} normal.
\end{example}

Hence, the evolution of the system may be unital even if $\A$ is not
normal.  Similarly, the superoperator $\A$ of an open system may be
normal even if the evolution is not unital, showing that the
condition $\sum_d [V_d,V_d^\dag]=0$ is not a necessary condition for
the HS distance between quantum states to be monotonically
decreasing.

\begin{figure}
\includegraphics[width=0.4\columnwidth]{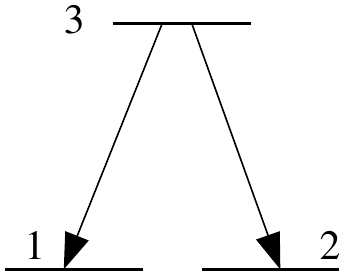}
\caption{Three-level $\Lambda$-system with (independent) decay from
the excited state to either of the two ground states.} \label{fig2}
\end{figure}

\begin{example}
\label{example:two}
Consider a three-level $\Lambda$ system with decay from the excited
state to the ground state as pictured in Fig.~\ref{fig2}. Neglecting
the Hamiltonian part, the evolution equation given by two
spontaneous emission processes, characterized by Lindblad operators
\begin{equation*}
  V_1 = \begin{pmatrix} 0 & 0 & 1\\ 0 & 0 &0 \\ 0 & 0 & 0
  \end{pmatrix},
  V_2 = \begin{pmatrix} 0 & 0 & 0\\ 0 & 0 &1 \\ 0 & 0 & 0 \end{pmatrix}.
\end{equation*}
It easy to check that $[V_1,V_1^\dag]=\diag(1,0,-1)$ and
$[V_2,V_2^\dag]=\diag(0,1,-1)$ and thus $\sum_d
[V_d,V_d^\dag]\neq0$. Thus the evolution is not unital, and neither
of the corresponding superoperators $\A_{V_1}$ and $\A_{V_1}$ are
normal, but if both decay processes are equally likely, the
off-diagonal terms in the sum $\A_{V_1}+\A_{V_1}$ cancel, and the
resulting superoperator $\A$ is diagonal, hence normal.  Thus the HS
distance between quantum states under this semi-group dynamics is
monotonically decreasing although the HS norm is not.
\end{example}

\section{Monotonicity of HS Distance for Non-normal $\A$, non-unital evolution}

We have shown that contractivity of the HS norm is a sufficient but
not a necessary requirement for contractivity of the HS distance,
and that normality of the dissipative part of $\A$~\footnote{The
Hamiltonian part of the dynamics does not contribute to the
symmetric part $\A+\A^T$ of $\A$ and thus is of no concern here.} is
an alternative sufficient condition for monotonicity of the HS
distance under Lindblad semi-group dynamics.  Finally, we show that
even both sufficient conditions together are not necessary, i.e.,
there are systems for which $\A$ is not normal and $\vec{c}\neq 0$,
but the HS distance is still monotontically decreasing.

\begin{figure}
\includegraphics[width=\columnwidth]{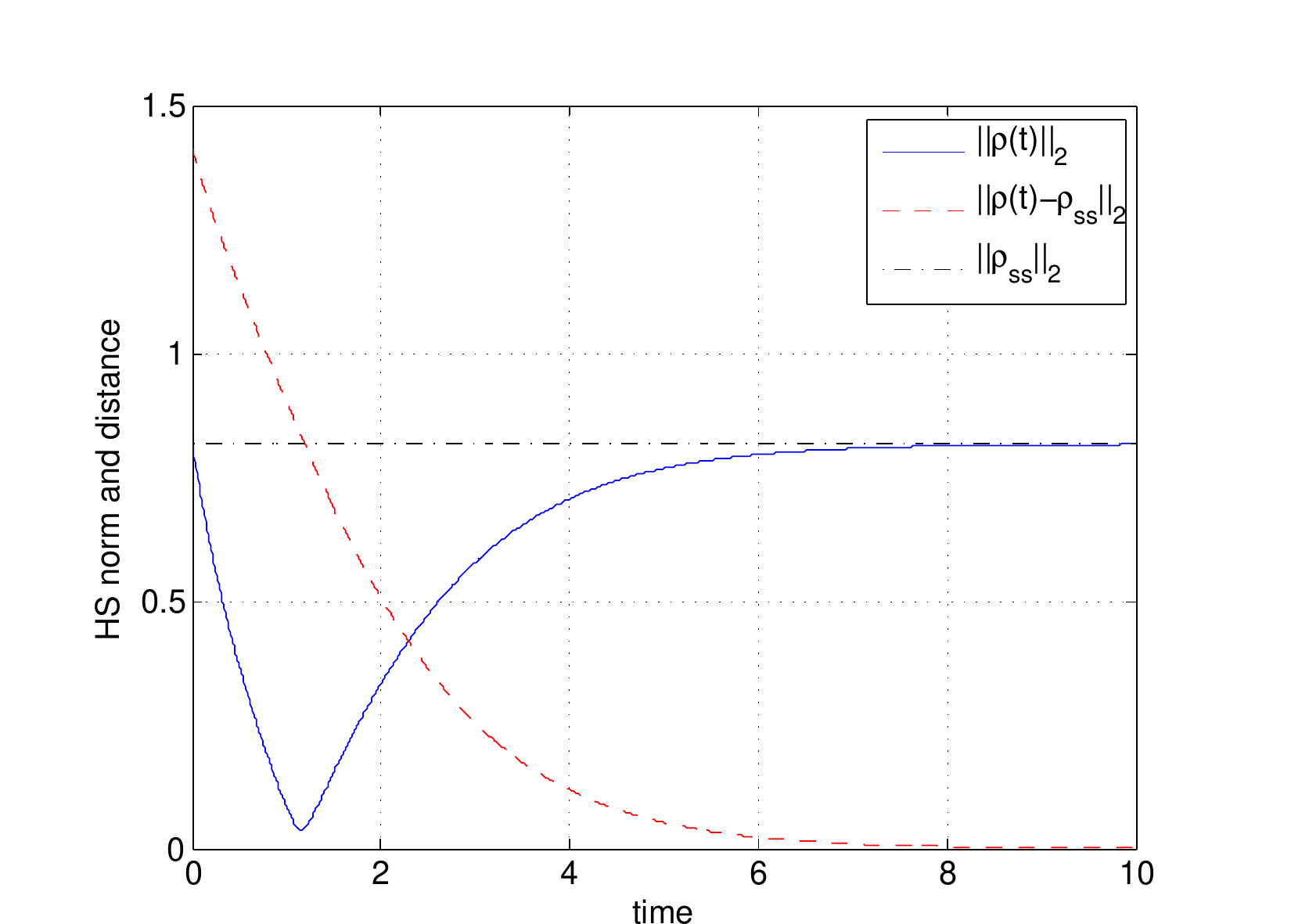}
\caption{Bloch superoperator $\A$ not normal but negative definite.
The evolution is not unital and $\norm{\rho(t)}_2$ for
$\rho_0=\diag(0,0,1)$ shows that the HS norm is not contractive.
The the distance of $\rho(t)$ from the steady state $\rho_{ss}$,
however, is monotonically decreasing, as is the distance between any
two quantum states.} \label{fig3}
\end{figure}

\begin{example}
Consider the purely dissipative semi-group dynamics given by
$\dot\rho(t)=\D[V]\rho(t)$ with
\begin{equation*}
  V = \begin{pmatrix} 0 & 1 & 0\\ 0 & 0 &1 \\ 0 & 0 & 0 \end{pmatrix}.
\end{equation*}
$V$ is clearly not normal and neither is the superoperator $\A$,
which, with respect to the standard basis~(\ref{eq:pauliN}) for the
trace-zero Hermitian matrices for $N=3$, takes the explicit form
\begin{equation*}
 \A = \begin{pmatrix}
        -1&0&0&0&\sqrt{\frac{4}{3}}&0&0&0\\
         0&-\frac{1}{2}&0&0&0&1&0&0\\
         0&0&-\frac{1}{2}&0&0&0&0&0\\
         0&0&0&-\frac{1}{2}&0&0&0&1\\
         0&0&0&0&-1&0&0&0\\
         0&0&0&0&0&-1&0&0\\
         0&0&0&0&0&0&-\frac{1}{2}&0\\
         0&0&0&0&0&0&0&-1
\end{pmatrix}.
\end{equation*}
Furthermore, we have $\vec{c}\neq 0$, i.e., the evolution is not
unital. Nonetheless, we can check that the eigenvalues of $\A+\A^T$
are $-2 \pm \frac{2}{3}\sqrt{3}$, $-\frac{3}{2}\pm
\frac{1}{2}\sqrt{5}$ and $-1$, all of which are negative.  (There
are only five distinct eigenvalues as the last three occur with
multiplicity $2$.)  Thus the HS distance between any two states is
monotonically decreasing, and we can easily verify that the system
has a unique steady state $\rho_{\ss}=\ket{1}\bra{1}$ at the
boundary of $\DD[\H]$.  The HS norm, of course, is not contractive.
E.g., if we start in the state $\rho_0 = \diag(0,0,1)$ then the HS
norm $\norm{\rho(t)}_2=\sqrt{\Tr(\rho(t)^2)}$ first decreases and
then increases as shown in Fig.~\ref{fig3}.
\end{example}

However, examples when $\A$ is not normal, the evolution not unital
and the HS norm non-contractive, but the HS distance is still
monotonically decreasing appear to be increasingly hard to find in
higher dimensions. Numerical tests with randomly generated simple
Lindblad generators $V$ suggest that when the dissipative part of
$\A$ is not normal then there is a high probability that the
symmetric part $\A+\A^T$ will have at least one positive eigenvalue.
Moreover, both the probability of a positive eigenvalue and the
number of positive eigenvalues of $\A+\A^T$ appear to increase with
the system dimension (See Table I).

\begin{table}
\begin{tabular}{|l|c|c|c|c|c|c|c|}
\hline
 $N$                 & 2   & 3 & 4 & 5 & 6 & 7 & 8 \\\hline
 non-contractive     & 0\% & 63.7\% & 95.4\% & 100\% & 100\% & 100\% & 100\%
\\\hline
 max. no. $\gamma>0$ & 0   & 1 & 3 & 4 & 7 & 9  &  11 \\\hline
\end{tabular}
\caption{Percentage of systems governed by the Lindblad master
equation $\dot\rho(t)=\D[V]\rho(t)$ with randomly generated $V$, for
which the HS distance is not monotonically decreasing and maximum
number of positive eigenvalues $\gamma$ of $\A+\A^T$ as a function
of the system dimension shows a rapid increase in the frequency of
non-contractive dynamics for increasing $N$.  For each $N$ 1000
systems were simulated.  All randomly generated $V$ (and hence $\A$)
were non-normal, as expected as normal matrices form a measure-zero
set.}  \label{table1}
\end{table}

\section{Conclusion}

Unlike the trace distance, and perhaps contrary to intuition, the HS
norm is generally not contractive under positive trace-preserving
maps, except for $N=2$.  For open systems governed by a Lindblad
master equation the necessary and sufficient conditions for
contractivity of the HS norm translate into a necessary and
sufficient condition for the Lindblad generators.  This condition is
also sufficient to ensure that the HS distance between quantum
states is monotonically decreasing, but it is not necessary.  We
derive alternative necessary and sufficient conditions for
monotonicity of the HS distance in terms of the spectrum of the
symmetric part of the super-operator, and show that they lead to
alternative sufficient conditions for monotonicity of the HS
distance, which are strictly weaker than those for monotonicity of
the HS norm. This means that the HS distance between any two quantum
states under Lindblad dynamics can be monotonically decreasing even
if the HS norm $\norm{\rho(t)}_2$ of quantum states is not monotonic
under this evolution.  Although the criteria for monotonicity of the
HS distance are weaker, in general it is not monotonically
decreasing if the Hilbert space dimension is greater than $2$, even
for systems that have a unique steady state.  The non-monotonicity
of the HS distance has important implications for, e.g., quantum
control, showing that unlike for Hamiltonian systems, it is not a
suitable Lyapunov function for generic open systems.  It is a
suitable candidate for a Lyapunov function only in very special
cases, e.g., when all the Lindblad decoherence operators are normal.
Although such systems are only a small subset of possible systems in
higher dimensions, they do include the important special case where
the decoherence is induced by measurement of a Hermitian observable.

\acknowledgments

XW is supported by a studentship from the Cambridge Overseas Trust
and an Elizabeth Cherry Major Scholarship from Hughes Hall,
Cambridge. SGS acknowledges funding from EPSRC Advanced Research
Program Grant RG44815, the EPSRC QIP Interdisciplinary Research
Collaboration (IRC) and Hitachi.  We sincerely thank Berry Groisman,
Francesco Ticozzi, Arieh Iserles, Ismail Akhalwaya, Pierre de
Fouquieres, and David Perez-Garcia for helpful comments and
discussions.

\appendix
\section{Standard Basis for Hermitian Matrices}
\label{app:A}

A standard orthonormal basis for the trace-zero Hermitian matrices
for any $N$ is given by $\{\sigma_k\}_{k=1}^{N^2-1}$ where
$\sigma_k=\sigma_{k(r,s)}$ with $k=r+(s-1)N$ and
\begin{subequations}
  \label{eq:pauliN}
  \begin{align}
    \sigma_{rs} &= \textstyle \frac{1}{\sqrt{2}}(\ket{r}\bra{s} +
\ket{s}\bra{r}) \\
    \sigma_{sr} &= \textstyle i\frac{1}{\sqrt{2}}(-\ket{r}\bra{s} +
\ket{s}\bra{r})\\
    \sigma_{rr}  &= \textstyle \frac{1}{\sqrt{r+r^2}}
          \left(\sum_{k=1}^r \ket{k}\bra{k} - r\ket{r+1}\bra{r+1} \right)
\end{align}
\end{subequations}
for $1\le r\le N-1$ and $r<s\le N$.

\section{Bloch representation}
\label{app:B}

Let $\{\sigma_k\}_{k=1}^{N^2}$ be a basis for the Hermitian
matrices. With respect to this basis the master
equation~(\ref{eq:LME}) with dissipation term~(\ref{eq:D}) can be
written in coordinate form as a linear matrix differential equation
(DE) $\dot{\vec{r}}=(\Lb+\sum_d \Db^{(d)})\vec{r}$, where
$\vec{r}=(r_n)\in\RR^{N^2}$ with $r_n=\Tr(\rho\sigma_n)$ and $\Lb$
and $\Db^{(d)}$ are $N^2\times N^2$ (real) matrices with entries
\begin{subequations}
\begin{align}
\label{eqn:LD}
  L_{mn}       &= \Tr(i H [\sigma_m,\sigma_n]) \\
  D_{mn}^{(d)} &= \Tr(V_d^\dag \sigma_m V_d \sigma_n)
                  -\frac{1}{2} \Tr(V_d^\dag V_d \{\sigma_m,\sigma_n\})
\end{align}
\end{subequations}
where $\{A,B\}=AB+BA$ is the usual anticommutator.

If we choose the basis such that
$\sigma_{N^2}=\frac{1}{\sqrt{N}}\ONE$ and the remaining basis
elements form a basis for the trace-zero Hermitian matrices, then,
noting that $r_N=\frac{1}{\sqrt{N}}\Tr(\rho)=\frac{1}{\sqrt{N}}$ is
constant and thus $\dot{r}_N=0$, we can define the reduced so-called
Bloch vector $\vec{s}=(r_1,\ldots,r_{N^2-1})^T$, and rewrite the
linear matrix DE for $\vec{r}$ as affine-linear matrix DE
$\dot{\vec{s}}(t) = \A \vec{s}(t) + \vec{c}$ for $\vec{s}$,
where $A$ is an $(N^2-1)\times (N^2-1)$ real matrix with $A_{mn}=
L_{mn}+ \sum_d D_{mn}^{(d)}$ and 
$c_m=\frac{1}{\sqrt{N}}\sum_d D_{mN}^{(d)}=\frac{1}{N}\sum_d\Tr([V_d,V_d^\dag]\sigma_m)$.


\begin{thebibliography}{99}
\bibitem{Nielson}
M. A. Nielson and I. L. Chuang, Quantum Computation and Quantum
Information. (Cambridge University Press, 2000), Theorem 9.2,
406-407.

\bibitem{Breuer}
H.-.P Breuer, F. Petruccione, The Theory of Open Quantum Systems.
(Oxford University Press, 2002), p 123.

\bibitem{Ruskai}
M. B. Ruskai, Rev. Math. Phys. \textbf{6}(5A), 1147-1161 (1994)

\bibitem{Vedral}
V. Vedral, M. B. Plenio, Phys. Rev. A \textbf{57}, 1619 (1998);

V. Vedral, M. B. Plenio, M. A. Rippin and P. L. Knight, Phys. Rev.
Lett.~\textbf{78}, 2275 (1997);

\bibitem{Witte}
C. Witte and M. Trucks, Phys. Lett. A \textbf{257}, 14 (1999)

\bibitem{Ozawa}
M. Ozawa. Phys. Lett. A~\textbf{268}, 158 (2000)

\bibitem{Garcia}
D. Perez-Garcia, M. M. Wolf, D. Petz, M. B. Ruskai, J. Math. Phys.
\textbf{47} 083506 (2006)

\bibitem{Mirrahimi2005}
M. Mirrahimi, G. Turinici, Automatica~\textbf{41}, 1987 (2005)

\bibitem{altafini}
C. Altafini, IEEE Trans. Autom. Control~\textbf{52}, 1 (2007)

\bibitem{Wang-Schirmer}
X. Wang, S. Schirmer, http://arXiv.org/abs/0801.0702

\bibitem{Ticozzi}
F. Ticozzi and L. Viola, IEEE Trans. Autom. Control \textbf{53}(9),
2048-2063 (2008); F. Ticozzi and L. Viola,
http://arXiv.org/abs/0809.0613.

%

\bibitem{Yosida}
Yosida K. Functional analysis. Springer, 1980.

\end{thebibliography}
\end{document}